\documentclass[11pt,letterpaper]{article}
\usepackage[margin=1in]{geometry}
\usepackage[hyphens]{url}
\usepackage{graphicx}
\usepackage{xcolor}
\usepackage{natbib}
\usepackage{caption}
\usepackage{microtype}
\usepackage[colorlinks=true,allcolors=green!50!black]{hyperref}
\hypersetup{
  pdftitle={A Constant Metric Distortion Protocol for Approval Voting Given Plurality Polls},
  pdfauthor={Fabian Frank and Jannik Peters}
}
\usepackage{doi}
\usepackage{amsmath}
\usepackage{amsfonts}
\usepackage{amssymb}
\usepackage{amsthm}
\usepackage{nicefrac}
\newtheorem{theorem}{Theorem}
\newtheorem{corollary}{Corollary}
\newtheorem{conjecture}{Conjecture}

\newtheorem{lemma}{Lemma}
\newtheorem{example}{Example}
\newtheorem{protocol}{Protocol}

\newtheorem{definition}{Definition}

\usepackage{tikz}
\usepackage{thm-restate}
\usepackage{mathtools}
\usetikzlibrary{calc, positioning, shapes.geometric}
\usepackage{pgfplots}
\pgfplotsset{compat=1.18} 

\DeclareMathOperator{\plu}{plu}

\DeclareMathOperator{\rank}{rank}
\DeclareMathOperator*{\argmax}{arg\,max}
\DeclareMathOperator*{\argmin}{arg\,min}
\DeclareMathOperator{\dist}{dist}
\DeclareMathOperator{\soc}{sc}

\usepackage{cleveref}
\crefname{protocol}{Protocol}{Protocols}
\crefname{corollary}{Corollary}{Corollaries}
\newcommand{\approved}[1]{\textcolor{red}{ #1}}

\renewcommand{\top}{\mathrm{top}}
\providecommand{\texorpdfstring}[2]{#1}

\allowdisplaybreaks

\title{A Constant Metric Distortion Protocol for Approval Voting Given Plurality Polls}
\author{Fabian Frank, Jannik Peters\footnote{E-Mail: \url{fabian_w.frank@tum.de}, \url{jannikpeters2512@gmail.com}}\\
Technische Universität München, Shanghai University of Finance and Economics}
\date{\today}

\begin{document}

\maketitle

\begin{abstract}
    Approval voting is a simple and well-regarded voting rule: voters submit approval ballots (subsets of the candidates) and the candidate receiving the most approvals wins. One major limitation of approval voting is that it is not clear which candidates voters should approve if they have an underlying strict order over the candidates. In this paper, we initiate the study of approval voting under a simple kind of information: plurality polls. That is, we assume that for each candidate we know the share of voters who rank this candidate as their top choice. Using these plurality polls, we suggest a simple protocol parameterized by a fraction $k \in (0,1)$: every voter should approve the smallest prefix of their preference list containing the first choices of at least a fraction $k$ of the voters. We evaluate this protocol via the framework of metric distortion and show that for the optimal choice of $k$, this rule achieves a metric distortion of $2 + \sqrt{5}  \simeq 4.236$. 
    The proof techniques we use for this statement also show that the Bucklin voting rule has a metric distortion of at most $5$. Finally, we evaluate the robustness of our protocol to noise and show that the upper bounds obtained for our protocol are tight. 
\end{abstract}

\section{Introduction}
Single-winner voting is one of the core topics of (computational) social choice and a key ingredient in collective decision-making. Voters submit preferences over candidates, which are then aggregated to select a single candidate. This selection process is usually handled via so-called \emph{voting rules}: functions that map elections to their winners. A big part of social choice theory is focused on evaluating the quality of these voting rules. A recently popularized framework for evaluating the \emph{utilitarian} performance of voting rules is \emph{metric distortion} \citep{ABE+18a, GHS20a}. 
This framework assumes that voters and candidates are both embedded in a common metric space and that the distance between a voter and a candidate  corresponds to the cost this voter incurs if this candidate is selected. It is further assumed that the preferences of voters are consistent with this metric space, i.e., if the distance from voter $i$ to candidate $x$ is smaller than the distance from $i$ to candidate $y$ then $i$ weakly prefers $x$ to $y$. The metric distortion of a candidate is then
the worst ratio, over all consistent metric spaces, between its social cost --- the
sum of the distances from all voters to it --- and the social cost of an optimal
candidate. 
A successful line of work has shown that there exist deterministic ordinal rules with worst-case metric distortion 3 \citep{GHS20a, KiKe22a}, and that no deterministic ordinal rule can do better \citep{ABE+18a}. In general, several of the voting rules that achieve a good metric distortion are desirable outside of the metric distortion context, for instance, the plurality veto rule \citep{KiKe22a, KiKe23a, Mou81a}, Copeland's rule \citep{Cope51a}, or Maximal Lotteries \citep{Krew65a, CRWW24a}. One frequently advocated rule that is largely missing from the literature on metric distortion though is \emph{approval voting} \citep{BrFi07a}. In approval voting, instead of submitting a ranking over candidates, voters submit a subset of candidates they approve of. 
While this leads to a voting rule that is simple to present and whose ballots are simple to aggregate, approval voting has one major downside that makes analyzing it difficult: there is no ground-truth approval ballot. Instead, any prefix of a voter's preferences could be a sincere ballot and the voter needs to decide which prefix to choose. This, despite the procedural simplicity of approval voting, adds a layer of difficulty to the voting rule. 
Simple strategies that ignore the preferences of other voters entirely (such as ``approve your top-$\ell$ candidates'') lead to unbounded metric distortion \citep{PiSk19a}. Hence, any approval voting protocol (with the goal of achieving a good metric distortion) must use at least some information about the preferences of other voters.
That approval voting inherently requires the voters to adopt strategies based on the preferences of other voters is well known in the economics literature. However, which information to use and how to adapt to it is unclear. One common strategy, suggested for instance by \citet{MyWe93a} or \citet{Lasl09a} (see also \citet[Chapter~7]{BrFi07a}), is that voters should try to identify the two \emph{front-runners} of the election and then approve their smallest prefix including one of the front-runners. While this can have favorable properties in equilibrium, such a model poses another problem: voters need to be able to accurately identify front-runners, which again requires voters to make assumptions about the preferences of other voters and their strategic behavior. 

In this work, we suggest using a very simple (and relatively readily available) type of information: \emph{plurality polls}. That is, we assume that before conducting the election, voters learn the distribution of top-choice preferences (i.e., how many voters rank which candidate first), for instance through a publicly conducted poll.
Using this information, we suggest using the following simple strategy parameterized by $k \in (0,1)$: each voter approves the smallest prefix of their ordinal preference such that at least $k \cdot n$\footnote{Here, $n$ is the number of voters.}  voters have their top-choice candidate in this set. We call this protocol the \emph{$k$-plurality approval} protocol. After all voters have submitted their approval ballots, the winners according to the protocol are all candidates with the most approvals. 
\begin{figure}[t]
    \centering
    \begin{tikzpicture}[scale=0.75, transform shape]
\begin{axis}[
    xlabel={$k$},
    ylabel={Distortion},
    xmin=0, xmax=1,
    ymin=0, ymax=10,
    domain=0.01:0.99, 
    samples=200,
    axis lines=left,
    legend style={
    at={(0.65,0.6)},
    anchor=north west,
    },
    grid=both,
    grid style={dashed, gray!30},
    restrict y to domain=0:12, 
    width=10cm,
    height=8cm
]

\addplot[blue, thick, domain=0.05:1] {(2-x)/x};
\addlegendentry{$f_1(k) = \frac{2-k}{k}$}

\addplot[red, thick, domain=0:0.95] {(3-x)/(1-x)};
\addlegendentry{$f_2(k) = \frac{3-k}{1-k}$}

\addplot[black, ultra thick, dashed] {max((2-x)/x, (3-x)/(1-x))};
\addlegendentry{$\max(f_1,f_2)$}

\filldraw[black] (axis cs:0.381966, 4.236068) circle (2pt);
\node[anchor=south] at (axis cs:0.382, 2.3) {$(0.382, 4.236)$};

\end{axis}
\end{tikzpicture}

\caption{Worst-case metric distortion of the $k$-plurality approval protocol.} 
\label{fig:distbound}

\end{figure}
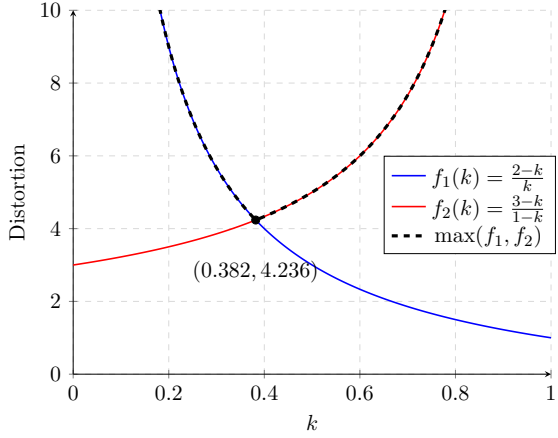
\subsection{Our Contribution}
We show that (provided that all voters truthfully follow the protocol and that the plurality scores are correctly elicited) the approval winner according to the $k$-plurality approval protocol has a metric distortion of at most $\max(\nicefrac{(2-k)}{k}, \nicefrac{(3-k)}{(1-k)})$ (see \Cref{fig:distbound} for a visualization of this bound) and hence is among the small group of voting rules that have a constant-factor metric distortion. In particular, this protocol has a metric distortion of $2 + \sqrt{5} \simeq 4.236$ for $k = \frac{3 - \sqrt{5}}{2} \simeq 0.382$ (thus equaling the distortion of the voting rule of \citet{MuWa19a}) and a metric distortion of $5$ for the majority threshold $k = \frac{1}{2}$.
We further show that these bounds are tight and that no anonymous approval voting based rule using plurality polls can have a better distortion than $2 + \sqrt{5}$. We then show that our analysis additionally extends to the fallback bargaining family of voting rules. This allows us to improve previous bounds of \citet{AFP22b} from $11$ to $5$ as well as to show that the first ``phase'' of the voting rule of \citet{KKK24a} (which is equivalent to Bucklin's voting rule)  already achieves a metric distortion of at most $5$.

Further, our analysis leads to one of the arguably simplest voting rules with a constant metric distortion so far: select a candidate minimizing the rank $r$ such that a weak majority of voters rank the candidate among their top $r$ alternatives. This voting rule is also known as the majoritarian compromise.   
Finally, turning back to $k$-plurality approval we study the impact of noisy plurality polls. We show that if run on a plurality poll that is close to the original plurality scores, then the metric distortion is also guaranteed to be close to the guarantees of the original protocol. Using this result we additionally obtain that it is sufficient to sample $\mathcal{O}(\log(nm/\delta)/\varepsilon^2)$ voters uniformly at random to achieve a metric distortion of $\max\left(\frac{2-k}{k}, \frac{3 - k - \varepsilon}{1 - k - \varepsilon}\right)$ with probability $1-\delta$. 

\subsection{Related Work} 
The general framework of distortion was introduced by \citet{ProRos06a} to measure the quantitative loss in social welfare by a voting rule that only has access to ordinal preferences. Our work in particular is situated in the framework of metric distortion introduced by \citet{ABE+18a}. \citeauthor{ABE+18a} showed that no deterministic voting rule can achieve a metric distortion of better than $3$, while Copeland's rule achieves a distortion of $5$. After intermediate improvements to $2+\sqrt{5}$
\citep{MuWa19a,Kem20a}, \citet{GHS20a} indeed presented a voting rule with distortion $3$, see also \citet{KiKe22a} and \citet{KiKe23a}. For randomized voting rules, it is currently known that the best achievable distortion lies between $2.1126$ \citep{ChRa22a} and $2.753$ \citep{CRWW24a} (with two recent unpublished preprints claiming an improvement to $2.5$ \citep{Fran26a, Ye26a}).

We are particularly related to recent research on the metric distortion of restricted classes of voting rules. For instance, \citet{AFP22b} study ``coordination dynamics'' in which voters need to coordinate in a decentralized fashion to obtain a winning candidate with low metric distortion. The coordination dynamic they suggest is equivalent to the majoritarian compromise rule and as shown by \citeauthor{AFP22b} achieves a metric distortion of $11$. Our work improves upon theirs in two ways: firstly, our $k$-plurality approval protocol can itself be understood as a two-round coordination dynamic achieving a better metric distortion than theirs, while secondly our proof template also applies to the majoritarian compromise rule and improves their upper bound from $11$ to $5$.

Another example is \citet{EHM24a}, who studied the query complexity of achieving low metric distortion, with one of their open questions being whether or not one can achieve a constant metric distortion using a protocol that only asks each voter $\mathcal{O}(m)$ pairwise comparisons (which cannot depend on the answers of the other voters). As we show later, our approval voting protocol can also be implemented using just $\mathcal{O}(m)$ pairwise comparisons per voter (at the cost of using two rounds of interaction). We are also related to \citet{AFP22a} who study various ``limited information'' settings such as truncated preferences or sampled voters.

Finally, by virtue of analyzing the metric distortion of approval voting we are closely related to \citet{PiSk19a} who analyze the distortion of approval voting in a setting in which each voter has an approval threshold (that is the voter approves a candidate if and only if the distance to this candidate is at most this threshold). They show that there always exists an approval radius for which approval voting has a distortion of $\frac{11}{3}$. This approval radius, however, cannot be identified using the ordinal preferences alone. \citeauthor{PiSk19a} further determine the distortion of approval voting based on the fraction of voters approving the optimal candidate and study what they call acceptability-based distortion which in essence measures how close a voting rule is to approval voting itself. Approval voting was also studied in the context of metric distortion by \citet{AFJV24a} who showed that using threshold approval queries (and if one knows the distances between the candidates), one can find a voting rule with a metric distortion of better than $3$.

\section{Preliminaries and Notation}  
Let  $N = \{1, \dots, n\} = [n]$ be the set of \emph{voters} and $C = \{c_1, \dots, c_m\}$ be the set of \emph{candidates}. Every voter $i \in N$ has a strict linear order $\succ_i \subseteq C\times C$
over the candidates. We denote the \emph{preference profile} by $\succ = (\succ_1, \dots, \succ_n)$.
Furthermore, for a voter $i \in N$ we write $a \succsim_i c$ if and only if $a \succ_i c$ or $a = c$. 
We denote by $\top(i)$ the most preferred candidate of $i$, i.e., $\top(i)\succsim_i c$ for all $c\in C$. For $c\in C$ we denote by $\plu(c) = \lvert\{i \in N \colon \top(i) = c\}\rvert$ the plurality score of a candidate and  by $p$ the normalized plurality score namely  $p(c) = \frac{\plu(c)}{n}$ for every $c \in C$. For a given $c \in C$ and voter $i \in N$ we write $\rank(i,c) \coloneqq \lvert \{c' \in C\colon c' \succ_i c\}\rvert + 1$. Note that $\rank(i, \top(i)) = 1$. Finally, given a voter $i \in N$ and a set of candidates $C' \subseteq C$ we say that $C'$ is a \emph{prefix} of $i$'s preference if $a \succ_i b$ for every $a \in C'$ and $b \in C \setminus C'$. 

For a given voter $i \in N$, an approval ballot of $i$ is a subset $A_i \subseteq C$. An approval profile $A = (A_1, \dots, A_n)$ is a collection of approval ballots. We also write
$A_c = \{i \in N \colon c \in A_i\}$ for the set of voters approving $c$ and call $\lvert A_c \rvert$ the \emph{approval score} of $c$. Given an approval profile $A$ the set of \emph{approval winners} is $\argmax_{c \in C} \lvert A_c\rvert$.

We assume that all voters and candidates lie in a common pseudo-metric space $(N \cup C,d)$ where $d \colon (N \cup C) \times (N \cup C) \rightarrow \mathbb{R}_{\geq 0}$ is a pseudo-metric. This means that for all $x,y,z \in (N \cup C)$ (i) $d(x,x) = 0$, (ii) $d(x,y) = d(y,x)$ and (iii) $d(x,z) \leq d(x,y) + d(y,z)$.

We say that a pseudo-metric $d$ is consistent with a profile $\succ$ if  $b \succ_i c$ implies that $d(i,b) \leq d(i,c)$ for all $i \in N$ and $b,c \in C$. Further, we denote by $\mathcal{D}(\succ)$ the set of all pseudo-metrics that are consistent with $\succ$.

We define the social cost of a candidate $c \in C$ given a pseudo-metric $d$ as $\soc(c,d) = \sum_{i \in N} d(i,c)$. If $d$ is clear from the context we omit it. Next, we define the metric distortion $\dist$ as 
$$\dist(c, \succ)=\sup_{d \in \mathcal{D}(\succ)} \frac{\soc(c,d)}{\min_{e \in C}\soc(e,d)}$$
with the convention that $\frac{0}{0} = 1$ and $\frac{x}{0} = \infty$ for $x > 0$.

We define a social choice correspondence as a function that maps every preference profile to a non-empty subset of the candidates in the profile.

The metric distortion of a social choice correspondence $f$ then is defined as
\[
\dist(f) \coloneqq \sup_\succ \max_{c \in f(\succ)} \dist(c,\succ),
\]
where the supremum ranges over all preference profiles, for any number of voters and candidates.

\section{Distortion Bounds for \texorpdfstring{$k$}{k}-Plurality Approval}

Our main object of study (and conceptual contribution) is the \emph{$k$-plurality approval protocol}. The protocol itself is parameterized by a value $k \in (0,1)$. It proceeds in two rounds: first, we elicit the top-choice candidate from each voter. Then every voter approves the smallest prefix of their ordinal preference that includes the top choices of at least $k\cdot n$ voters. We then elect the approval winners according to the resulting approval profile.

\begin{protocol}[$k$-Plurality Approval] \label{protoGeneral}

Let $k \in (0,1)$. Then, in the \emph{$k$-plurality approval} protocol we:
\begin{enumerate}
    \item For each voter elicit their top-choice candidate and compute the plurality score for each candidate.
    \item For each voter $i$ elicit their approval ballot $A_i$ which is the smallest prefix of $i$'s ballot with 
\(
\sum_{c \in A_i} \plu(c) 
\geq k n. \footnote{Such a prefix exists since all candidates combined have a total plurality score of $n$.}
\) Compile this into an approval profile $A$.
\item The winners according to the protocol are all approval winners according to $A$.
\end{enumerate}
\end{protocol}
\begin{example}
    To illustrate \Cref{protoGeneral} consider the preference profile depicted in \Cref{Example:Protocol}. On the left is the profile itself while on the right is the cumulative plurality scores of the prefix of each ballot. Let $k = \frac{1}{2}$. Then, the threshold is $5 \cdot \frac{1}{2} = 2.5$ and thus  each voter would approve the smallest prefix containing the first choices of at least $3$ voters. For our instance, this leads to the following approval ballots: 
        $A_1= \{a,c\}$, $A_2= \{a,d\}$, $A_3= \{b,a\}$, $A_4= \{c,b,e,d\}$, $A_5 = \{d,b,a\}$ (also marked in red in the profile). Thus, in this instance $a$ receives $4$ approvals and is the unique approval winner.
\end{example}

\begin{table}
    \centering
        \begin{tabular}[t]{ccccc}
            $1$  & $2$ & $3$ & $4$ & $5$\\
            \hline
            \approved{$a$} & \approved{$a$} & \approved{$b$} & \approved{$c$} & \approved{$d$}  \\
            \approved{$c$} & \approved{$d$} & \approved{$a$} & \approved{$b$} & \approved{$b$}\\
            $b$ & $b$ & $c$ & \approved{$e$} & \approved{$a$}\\
            $d$  & $c$ & $e$ & \approved{$d$} & $e$ \\
            $e$  & $e$ & $d$  & $a$ & $c$\\
        \end{tabular} \quad \quad \quad
        \begin{tabular}[t]{ccccc}
            $1$  & $2$ & $3$ & $4$ & $5$\\
            \hline
            \approved{$2$} & \approved{$2$} & \approved{$1$} & \approved{$1$} & \approved{$1$}  \\
            \approved{$3$} & \approved{$3$} & \approved{$3$} & \approved{$2$} & \approved{$2$} \\
            $4$ & $4$ & $4$ & \approved{$2$} & \approved{$4$}\\
            $5$  & $5$ & $4$ & \approved{$3$} & $4$ \\
            $5$  & $5$ & $5$  & $5$ & $5$\\
        \end{tabular}
        \caption{Example ballots for \Cref{protoGeneral}.
        The left-hand side depicts the preference profile while the right-hand side denotes the cumulative plurality score of the prefix. Candidates marked in red correspond to the approval ballot of the corresponding voter.}
        \label{Example:Protocol}
\end{table}

Before we turn to our distortion bounds, we show an interesting auxiliary result about $k$-plurality approval. One possible measure of the complexity of a voting rule is the number of \emph{pairwise comparisons} one would need to ask each voter in order to be able to compute the voting rule \citep{EHM24a}. Many standard voting rules need to sort the entire preference list of each voter and thus require  $\Theta (m \log m)$ many comparisons. In comparison, we show that for \Cref{protoGeneral} it is sufficient to ask $\mathcal O(m)$ many pairwise comparisons per voter in order to compute its output.
\begin{restatable}{theorem}{pairwise}
    For every $k \in (0,1)$ it is sufficient to make $\mathcal O(m)$ many pairwise comparisons for every voter in order to determine the winners according to the $k$-plurality approval protocol. \label{thm:query_complexity}
\end{restatable}
\begin{proof}
    In order to determine the most preferred alternative and thus the vote for the plurality poll it is sufficient to go once over all alternatives. Thus, at most $m-1$ comparisons are necessary. 

    \citet{reiser1978linear} has shown that determining a maximal prefix whose weighted score is at most $r$ for any value $r$ can be done in $O(m)$. The algorithm they describe to solve this problem requires that all weights are positive. 

    To compute the approval ballots we want to determine the least preferred candidate that still gets approved.
    For this we run this algorithm  with parameter $r = \lceil kn \rceil$ on all candidates that have a positive plurality score. If the returned prefix $S$ then has a combined plurality score of $\lceil kn \rceil$ we determine the least preferred element in $O(m)$. Otherwise, we determine the most preferred candidate from all candidates with positive plurality score not in $S$ (again in $O(m)$).

We claim that the selected candidate by this method actually
 is the least preferred candidate $c$ that still gets approved. First, observe $c$ has to have a positive plurality score since otherwise it does not affect the plurality score of the prefix contradicting minimality.
Second, if $S$ already has a combined plurality score of precisely $\lceil kn \rceil$ since every candidate in $S$ has a positive plurality score we get that the least preferred candidate in $S$ is equal to $c$. Finally, if $\plu(S) < \lceil kn \rceil$ we know by the maximality of $S$ that any longer prefix $T \supset S$ only containing candidates with plurality score at least one satisfies $\plu(T) > \lceil kn \rceil$. Thus, $c$ corresponds to the most preferred candidate with positive plurality score outside $S$.

To now determine the approval ballots we can simply compare every other candidate with $c$.
Since $c$ is the least preferred approved one, every candidate preferred to $c$ has to be approved and every other candidate not.

In total this describes a procedure that requires $O(m)$ comparisons.
\end{proof} 

\subsection{Distortion Upper Bound}
For the remainder of this subsection, fix a preference profile $\succ$ and corresponding approval profile $A$ computed in \Cref{protoGeneral}, an arbitrary approval winner $a\in C$, a candidate $b\in C$, and a pseudo-metric $d\in\mathcal D(\succ)$.
We begin with a simple lower bound on the approval score of the winning candidate $a$.
\begin{lemma}
The approval winner $a$ has an approval score of at least $kn$.
    \label{lem:appr_lowerbound}
\end{lemma}
\begin{proof}
    By the definition of the protocol, for every voter $i \in N$ the sum of the plurality scores of the candidates in $A_i$ is at least $kn$ and hence $\sum_{i \in N}\sum_{c \in A_i}\plu(c) \ge n \cdot (kn) = k n^2$. Therefore, 
    $kn^2 \le \sum_{i \in N}\sum_{c \in A_i}\plu(c) = \sum_{c \in C} \lvert A_c\rvert \plu(c) \le \sum_{c \in C} \lvert A_{a}\rvert \plu(c) = \lvert A_{a}\rvert \sum_{c \in C} \plu(c) = \lvert A_{a}\rvert \cdot n.$ Hence, $\lvert A_{a}\rvert \ge kn$.  
\end{proof}

In general, to upper bound the distortion of our protocol, our goal is to construct a fractional assignment between the voters in such a way that we can let some voters ``pay'' for the costs incurred by other voters they are assigned to. Our key idea for this assignment follows along the idea of \citet{MuWa19a} and their matching uncovered set (we will, however, have to match voters multiple times). Our goal is to bound the social cost of the approval winner $a$ by a factor of the social cost of any candidate $b$. 

 Given $a$ and $b$ and two voters $i, j \in N$ we say that $i$ \emph{dominates} $j$ if there exists a candidate $c \in C$ with $a \succsim_i c$ and $c \succsim_j b$ \citep[see also][]{MuWa19a}.  Using this domination relation, we get two simple upper bounds on the social costs incurred by voters $i$ and $j$ respectively.

\begin{lemma}\label{lemmaC}
    Let $i, j \in N$ such that $i$ dominates $j$.
    Then $d(i,a)\le d(i,b) + 2\,d(j,b)$.
\end{lemma}
\begin{proof}
    Let $c$ be the candidate witnessing the domination relation. First, we know that $d(i,a) \le d(i,c)$ and $d(j,c)\le d(j,b)$. By the triangle inequality we get $d(b,c) \le d(j,b) + d(j,c)\le 2\,d(j,b)$
and thus
    $d(i,a)\le d(i,c)\le d(i,b)+d(b,c)\le d(i,b)+2\,d(j,b).$
\end{proof}
Similarly, we can upper bound the cost of voter $j$.
\begin{lemma} \label{lemmaE}
    Let $i, j \in N$ such that $i$ dominates $j$.
    Then $d(j,a)\le 2\,d(i,b) +  3\,d(j,b)$.
\end{lemma}
\begin{proof}
    Using \Cref{lemmaC} and the triangle inequality we get
    $d(a,b) \le d(i,a) + d(i,b) \le \big(d(i,b) + 2\,d(j,b)\big) + d(i,b) = 2\,d(i,b) + 2\,d(j,b)$ and thus
    $d(j,a) \le d(j,b) + d(a,b) \le d(j,b) + \big(2\,d(i,b) + 2\,d(j,b)\big) = 2\,d(i,b) + 3\,d(j,b).$
\end{proof}

Using these upper bounds on the respective social costs, our target now becomes to strategically pair voters with other voters they dominate. To construct this assignment, we begin by defining a partition of $N$ based on whether the voters approve $a$, $b$, both, or neither.

\begin{definition}\label{def:voter_partitions}
    Given the approval winner $a$ and an arbitrary alternative $b$, we partition the voters $N$ into four disjoint subsets based on their approval set $A_i$:
    \begin{itemize}
        \item $N_{ab} = \{i \in N : a \in A_i \land b \in A_i\}$
        \item $N_{a} = \{i \in N : a \in A_i \land b \notin A_i\}$
        \item $N_{b} = \{i \in N : b \in A_i \land a \notin A_i\}$
        \item $N_{\emptyset} = \{i \in N : a \notin A_i \land b \notin A_i\}$
    \end{itemize}
\end{definition}

Note that $N_{ab}, N_a, N_b$, and $N_\emptyset$ form a partition of $N$ (with potentially empty sets). We observe two simple facts about our partition sets. Firstly, by \Cref{lem:appr_lowerbound} it must hold that $\lvert N_{ab} \cup N_a\rvert \ge \lceil kn\rceil$. Secondly, since $a$ is an approval winner it has to hold that  $\lvert N_a\rvert \geq \lvert N_b \rvert$. 

Using these four sets, we begin to reconstruct our domination relation. First, we observe that every voter approving $a$ dominates every voter in $N_\emptyset$.
 
\begin{lemma}
    Every $i \in N_{ab}\cup N_a$ dominates every $j \in N_\emptyset$. \label{lemmaD}
\end{lemma}
\begin{proof}
    Let $L_i = \{x \in C\colon a \succsim_i x\}$ be the set of candidates that voter $i$ weakly ranks below $a$. Because $i$ approves $a$, the plurality threshold is not met prior to candidate $a$. From this we get $\sum_{x \in C \setminus L_i} \plu(x) < kn$ and therefore $\sum_{x \in L_i} \plu(x)  > (1-k)n$.
    Further, let $A_j$ be the set of candidates approved by voter $j$. By the definition of the protocol, $\sum_{x \in A_j} \plu(x) \ge kn$. Thus, since the sum over all plurality scores is $n$ there must exist a candidate $c \in L_i \cap A_j$ via which $i$ dominates $j$. As $a\succsim_i c$ and $c \succ_j b$ (since voter $j$ does not approve $b$) we get that $i$ dominates $j$.
\end{proof}
 Further, we can show that each voter in $N_{ab}$ must dominate at least $(1-k)n$ voters. 
\begin{lemma}
    For every voter $i \in N_{ab}$ there exist at least $(1-k)n $ voters dominated by $i$. \label{lemmaOutdegree}
\end{lemma}
\begin{proof}
    Again, we know that the candidates $L_i = \{x \in C  \colon a \succsim_i x\}$  have total plurality score at least $(1-k)n$. The plurality score counts the voters $j$ who top-rank one of these candidates, and for such a $j$ we have $a\succsim_i \top(j)$. There are therefore at least $(1-k)n$ such voters $j$ and with $c=\top(j)$ we get both $a\succsim_i c$ and $c\succsim_j b$. Hence, $i$ dominates at least $(1-k)n$ different voters.
\end{proof}

Finally, every voter in $N_b$ dominates every voter in $N_a$. 
\begin{lemma}
    Let $i \in N_b$ and $j \in N_a$ then $i$ dominates $j$. \label{lem:b_out}
\end{lemma}
\begin{proof}
    This follows from $a \succsim_i a$ and $a \succsim_j b$. 
\end{proof}

 Now, we have all the ingredients to construct our fractional assignment. We will phrase it here in a more general form that will later allow us to apply the same proof to other voting rules. Formally, for us a fractional assignment is a function $w \colon (N_b \cup N_{ab} \cup N_\emptyset) \times N \to [0,1]$. The interpretation of this fractional assignment will be that $w(i,j)$ is (roughly) the weight by which $j$ compensates the social cost of $i$. We do not assign voters from $N_a$ (at least not on the left side of the assignment) as they are already guaranteed to prefer $a$ to $b$ by virtue of approving $a$ but not $b$. Now, for our fractional assignment, let $t \ge 0$ be a parameter. With this parameter we require five properties:
\begin{itemize}
    \item[i.)] $\sum_{j \in N} w(i,j) = 1$ for every $i \in N_b \cup N_{ab} \cup N_\emptyset$;
    \item[ii.)] $\sum_{i \in N_b \cup N_{ab} \cup N_\emptyset} w(i,j) \le t + 1$ for all $j \in N \setminus N_\emptyset$;
    \item[iii.)] $\sum_{i \in N_b \cup N_{ab} \cup N_\emptyset} w(i,j) \le t$ for all $j \in N_\emptyset$;
    \item[iv.)] If $w(i,j) > 0$ and $i \notin N_\emptyset$ then $i$ dominates $j$;
    \item[v.)] If $w(i,j) > 0$ and $i \in N_\emptyset$ then $j$ dominates $i$. 
\end{itemize}
For each $i \in N_b \cup N_{ab} \cup N_\emptyset$ we write $w(i)$ to be the set of voters $j \in N$ with $w(i,j) > 0$.
We can show that the existence of such an assignment is indeed enough to certify that $a$ has distortion at most $3 + 2t$.
\begin{restatable}{lemma}{lemupper}\label{thm:upperBound}
    Let $t\ge 0$. If a fractional assignment $w$ satisfying conditions i) -- v) exists for parameter $t$ and if every voter in $N_a$ prefers $a$ to $b$, then $\soc(a) \le (3 + 2t)\soc(b)$.
\end{restatable}
\begin{proof}
    We get $\soc(a) = \sum_{i \in N} d(i,a) = \sum_{i \in N_a} d(i,a) + \sum_{i \in N_\emptyset} d(i,a) + \sum_{i \notin N_\emptyset \cup N_a} d(i,a)$. 

    First, for $i \in N_\emptyset$ we know that every voter in $w(i)$ dominates $i$ and hence, by \Cref{lemmaE},
    \begin{align*}
        d(i,a)
        \le \sum_{j \in w(i)} w(i,j)
        \bigl(3d(i,b) + 2d(j,b)\bigr) = 3d(i,b) + \sum_{j \in w(i)} w(i,j)2d(j,b).
    \end{align*}

    Second, for $i \notin N_\emptyset \cup N_a$ we know that $i$ dominates every voter in $w(i)$ and hence, by \Cref{lemmaC},
    \begin{align*}
        d(i,a)
        &\le \sum_{j \in w(i)} w(i,j)
        \bigl(d(i,b) + 2d(j,b)\bigr)= d(i,b) + \sum_{j \in w(i)} w(i,j)2d(j,b).
    \end{align*}

    Piecing these two inequalities together we get 
    \begin{align*}
        \soc(a)
        &= \sum_{i \in N} d(i,a) \\
        &= \sum_{i \in N_a} d(i,a)
          + \sum_{i \in N_\emptyset} d(i,a) + \sum_{i \notin N_\emptyset \cup N_a} d(i,a) \\
        &\le \sum_{i \in N_a} d(i,b) + \sum_{i \in N_\emptyset} 
          \left(3d(i,b) + \sum_{j \in w(i)} w(i,j)2d(j,b)\right) + \sum_{i \notin N_\emptyset \cup N_a} \left( d(i,b)
          + \sum_{j \in w(i)} w(i,j)2d(j,b) \right) \\
        &\le \sum_{i \in N_a} \left(d(i,b) + (t+1)\cdot 2d(i,b) \right) + \sum_{i \in N_\emptyset} \left( 3d(i,b) + 2t d(i,b) \right) + \sum_{i \notin N_\emptyset \cup N_a}
        \left(d(i,b)
          + (t+1)\cdot 2d(i,b) \right) \\
        &= (3+2t)\soc(b).
    \end{align*}
    The second to last step followed since every $i \in N_\emptyset$ appears at most at a fraction of $t$ and every $i \notin N_\emptyset$ appears at most at a fraction of $t+1$ of the ``right side'' of $w$.
\end{proof}

Now all that remains is to construct the fractional assignment $w$. We will show that such a fractional assignment exists via a simple application of Hall's Theorem, assuming the domination and size constraints from earlier. We will again formulate the lemma slightly more general than needed to accommodate later proofs. 
\begin{restatable}{lemma}{fractional} \label{lem:fractional_hall}
    Let $N_a, N_b, N_{ab}, N_\emptyset$ be a partition of $N$, $k \in (0,1)$, and $\delta \in [0,1-k)$. If 
    \begin{itemize}
        \item $\lvert N_a \cup N_{ab}\rvert \ge kn $;
        \item $\lvert N_a\rvert \ge \lvert N_b\rvert$;
        \item every voter in $N_a\cup N_{ab}$ dominates every voter in $N_\emptyset$;
        \item every voter in $N_{ab}$ dominates at least $(1-k-\delta)n$ voters;
        \item and every voter in $N_b$ dominates every voter in $N_a$.
    \end{itemize}
    There exists a fractional assignment $w$ satisfying conditions i) -- v) with parameter
    \[
    t_\delta \coloneqq \max\left\{
        \frac{1-2k}{k},
        \frac{k+\delta}{1-k-\delta}
    \right\}.
    \]
\end{restatable}
\begin{proof}
    To construct such an assignment it is sufficient to fractionally assign:
    \begin{itemize}
        \item voters in $N_{ab}$ to one of their at least $(1-k-\delta)n$ voters they dominate;
        \item voters in $N_b$ to voters in $N_a$;
        \item voters in $N_\emptyset$ to voters in $N_a \cup N_{ab}$.
    \end{itemize}
    Let $N' \subseteq N_{ab} \cup N_b \cup N_\emptyset$ be an arbitrary non-empty subset of voters. To apply Hall's Theorem to it remains to show that this subset collectively has $\lvert N'\rvert$ weighted voters they could be assigned to. We will do a case distinction based on the different sets the voters could belong to.

    First, if $N' \subseteq N_b$, we know that the voters in $N'$ could be fractionally assigned to $t_\delta+1 \ge 1$ copies of $N_a$. Since $\lvert N_a\rvert \ge \lvert N_b\rvert$ the Hall condition follows.

    Second, if $N' \subseteq N_b \cup N_\emptyset$ with at least one voter in $N'$ belonging to $N_\emptyset$, we know that the voters can be assigned to the $(t_\delta+1)$ copies of $N_a \cup N_{ab}$. Since $\lvert N_a \cup N_{ab}\rvert \ge kn$ we get that these are at least
    \begin{align*}
        (t_\delta+1)kn
        &= \left(\max\left(\frac{1-2k}{k},\frac{k + \delta}{1-k - \delta}\right)+1\right)kn \\
        &\ge\left(\frac{1-2k}{k}+1\right)kn = (1-k)n \ge \lvert N_b \cup N_\emptyset\rvert \ge \lvert N'\rvert
    \end{align*}
    many copies.

    Third, if $N' \subseteq N_b \cup N_{ab}$ with at least one voter in $N'$ belonging to $N_{ab}$, we know that the voter in $N_{ab}$ already dominates at least $(1-k - \delta)n$ many voters. Of these dominated voters, the voters in $N_\emptyset$ contribute $t_\delta$ while the voters outside $N_\emptyset$ contribute $(t_\delta+1)$. Hence, we get that these are at least \begin{align*}
        &t_\delta \lvert N_\emptyset\rvert
          + (t_\delta+1)\bigl((1-k - \delta)n-\lvert N_\emptyset\rvert\bigr) \\
        &\quad = (t_\delta+1)(1-k-\delta)n-\lvert N_\emptyset\rvert \\
        &\quad = \left(\max\left(\frac{1-2k}{k},\frac{k+\delta}{1-k - \delta}\right)+1\right)
          (1-k - \delta)n-\lvert N_\emptyset\rvert \\
        &\quad \ge n-\lvert N_\emptyset\rvert \ge \lvert N_{ab}\rvert+\lvert N_b\rvert \ge \lvert N'\rvert .
    \end{align*}
    
   Finally, assume that $N'\subseteq N$ with there existing voters from both $N_{ab}$ and $N_\emptyset$ in $N'$. Hence, $N'$ can be assigned to all of $N_\emptyset \cup N_a\cup N_{ab}$. For this case, we need one property of $t$, namely $\max\left(\frac{1-2k}{k}, \frac{k+\delta}{1-k-\delta}\right) = t \ge 1-k$. To see this, we observe that $\frac{1-2k}{k} \ge 1-k$ if and only if $k^2 - 3k+1 \ge 0$ and $\frac{k}{1-k} \ge 1-k$ if and only if $k^2 - 3k+1 \le 0$. Since $\frac{k+\delta}{1-k-\delta} \ge \frac{k}{1-k}$ (due to $f(x) = \frac{x}{1-x}$ increasing in $(0,1)$) the inequality holds in either of the two cases. Now,  by using that $\lvert N_\emptyset\rvert \le (1-k)n$ and $\lvert N_a\cup N_{ab}\rvert \ge kn$ we get a fractional weight of at least 
   \begin{align*}
       t &\lvert N_\emptyset\rvert + (t+1)(\lvert N_a\cup N_{ab}\rvert) \ge \lvert N_a\cup N_{ab}\rvert + t\lvert N_\emptyset \cup  N_a\cup N_{ab} \rvert \\ &= \lvert N_a\cup N_{ab}\rvert + \max\left(\frac{1-2k}{k}, \frac{k + \delta}{1-k - \delta}\right)\lvert N_\emptyset \cup  N_a\cup N_{ab} \rvert \\
       & \ge \lvert N_a\cup N_{ab}\rvert + (1-k) (\lvert N_\emptyset\rvert + kn)
       \\ & \ge \lvert N_a\cup N_{ab}\rvert + (1-k)(\lvert N_\emptyset\rvert + \frac{k}{(1-k)}\lvert N_\emptyset\rvert) \\ 
       & \ge \lvert N_a\cup N_{ab}\rvert + (1-k) \left(\frac{1}{(1-k)}\lvert N_\emptyset\rvert\right) \\ 
       &\ge \lvert N_a\cup N_{ab} \cup N_{\emptyset}\rvert \ge \lvert N_b\cup N_{ab} \cup N_{\emptyset}\rvert \ge \lvert N'\rvert.
   \end{align*}

   Hence, our desired fractional assignment exists.
\end{proof}
These two lemmas now allow us to construct our upper bound on the distortion of the $k$-plurality approval winners. 
\begin{theorem} \label{theoremBound}
    For every $k \in (0,1)$ every winner according to $k$-plurality approval has a distortion of at most  $\max\left(\frac{2-k}{k}, \frac{3-k}{1-k}\right)$.
\end{theorem}
\begin{proof}
    To show this, we will apply \Cref{lem:fractional_hall} with $\delta = 0$. We indeed observe that our previously constructed partition into $N_a, N_b, N_{ab}, N_\emptyset$ satisfies all conditions of \Cref{lem:fractional_hall}. Hence, such an assignment exists with $t = t_0 = \max\left\{
        \frac{1-2k}{k},
        \frac{k}{1-k}
    \right\}$ and therefore \Cref{thm:upperBound} shows that the approval winner $a$ has a distortion of at most $3 + 2t = 3+2\max\left\{
        \frac{1-2k}{k},
        \frac{k}{1-k}
    \right\} = \max\left(\frac{2-k}{k}, \frac{3-k}{1-k}\right)$.
\end{proof}
This upper bound is minimized for $k = \frac{3 - \sqrt{5}}{2}$ as the following corollary shows. 
\begin{corollary}
    For $k = \frac{3 - \sqrt{5}}{2}$ the distortion of $k$-plurality approval is at most $2 + \sqrt{5} \simeq 4.236$.
\end{corollary}
\begin{proof}
    This follows since for $k = \frac{3 - \sqrt{5}}{2}$ we have $\frac{2-k}{k} = 2 + \sqrt{5} = \frac{3-k}{1-k}$. 
\end{proof}

\subsection{Fallback Bargaining Rules}
As a second application of our construction, we consider the class of $q$-approval fallback bargaining rules \citep{BrKi01a}. 
\begin{definition}
    Let $q \in [n]$. Then the \emph{$q$-approval fallback bargaining rule} selects all candidates in 
    \[
    \argmin_{c \in C} \min \{r \in [m] \colon  \lvert \{i \in N \colon \rank(i,c) \le r\}\rvert \ge q \}.
    \]
\end{definition}
In other words, $q$-approval fallback bargaining selects all candidates minimizing the rank $r$ needed after which at least $q$ voters have the candidate in their top-$r$ choices. As a short example, again consider the instance depicted in \Cref{Example:Protocol}. In this instance, for $q = \lceil \frac{n}{2}\rceil = 3$ both $a$ and $b$ are winners as they reach the required three voters for $r = 2$ which is the smallest possible.  

A particular special case for $q = \lceil\frac{n}{2}\rceil$ is the majoritarian compromise \citep{HuSe99a}, while for the strict majority threshold $q = \lfloor \frac{n}{2}\rfloor+1$ the rule is also known as the (simplified) Bucklin's rule.\footnote{It is also common to additionally restrict the voting rule to the set of candidates with the most voters ranking them within their top $r$. As this only refines the set of winners our later upper bound on the metric distortion automatically transfers.} In the context of metric distortion the majoritarian compromise was (essentially) rediscovered by \citet{AFP22b} who showed that it has a metric distortion of at most $11$ as well as by \citet{KKK24a} who designed a slightly more involved version of the majoritarian compromise\footnote{Instead of just selecting the majoritarian compromise winner, they additionally ran a head-to-head election with a second candidate representing the voters not part of the majority.} which they showed achieves a metric distortion of at most $44$. In particular, \citet{KKK24a} highlighted that their rule is arguably simpler than any other rule with constant metric distortion discovered before. 

Here, we show that for $k \in (0,1)$ the $\lceil kn \rceil$-approval fallback bargaining rule attains the same metric distortion guarantees as $k$-plurality approval. In particular, the metric distortion of the majoritarian compromise is at most $5$, thereby improving the bound of \citet{AFP22b}. This further shows that there is an arguably even simpler rule than the one of \citet{KKK24a} achieving a (better) constant metric distortion: select the candidates who are ranked by a (weak) majority of voters among their top-$r$ ranks for the smallest such $r$. 
\begin{restatable}{theorem}{fallbackupper}
    Let $n \in \mathbb{N}$, $q \in [n]$, and $k \in (0,1)$ with $q-1 \leq kn\leq q$. Then, the distortion of $q$-approval fallback bargaining on instances with $n$ voters is at most $\max\left(\frac{2-k}{k}, \frac{3-k}{1-k}\right)$. 
\end{restatable}
\begin{proof}
    Let $a \in C$ be a winner according to $q$-approval fallback bargaining and $b \in C$ be any other candidate. Further, let $r \in [m]$ be the first rank for which $\lvert \{i \in N \colon \rank(i,a) \le r\}\rvert \ge q$. 
    We again partition our set of voters into four sets. For this proof, we partition based on which rank the voters assign to $a$ and $b$:
    \begin{itemize}
        \item $N_{ab} = \{i \in N : \rank(i,a) \le r \land \rank(i,b) \le r-1\}$
        \item $N_{a} = \{i \in N : \rank(i,a) \le r \land \rank(i,b) > r-1\}$
        \item $N_{b} = \{i \in N : \rank(i,a) > r \land \rank(i,b) \le r-1\}$
        \item $N_{\emptyset} = \{i \in N : \rank(i,a) > r \land \rank(i,b) > r-1\}$
    \end{itemize}
    Again these sets form a partition (with potentially empty partition sets) of the voters. To apply \Cref{lem:fractional_hall} we need to verify that these sets satisfy the same domination properties and size properties. 

    First, by definition, we know that $\lvert N_a \cup N_{ab}\rvert \ge q \ge kn$. Secondly, it must hold that $\lvert N_a \rvert \ge \lvert N_b\rvert$. Otherwise, $\lvert N_{ab} \cup N_b \rvert \ge q$ and therefore candidate $b$ would have been a winner at rank $r-1$.

    \begin{lemma}
    Every $i \in N_{ab}\cup N_a$ dominates every $j \in N_\emptyset \cup N_a$. 
    \end{lemma}
    \begin{proof}
        Let $U_i = \{x \in C\colon a \succsim_i x\}$ be the set of candidates that voter $i$ ranks weakly below $a$. This set contains at least $m - r + 1$ candidates. Further, the set $U_j = \{x \in C\colon x \succsim_j b\}$ contains at least $r$ elements. Hence, there must exist a candidate $c \in U_i \cap U_j$ for which $a \succsim_i c$ and $c \succsim_j b$ holds.
    \end{proof}
     Thus, since $\lvert N_{ab} \cup N_b\rvert \le q-1\le kn$, it follows that every $i \in N_{ab}\cup N_a$ dominates at least $(1-k)n$ voters. 
    
    Finally, every voter in $N_b$ dominates every voter in $N_a$. 
    \begin{lemma}
        Let $i \in N_b$ and $j \in N_a$ then $i$ dominates $j$.
    \end{lemma}
    \begin{proof}
        This again follows by $a \succsim_i a$ and $a \succsim_j b$. 
    \end{proof}

    Therefore, our partition has the same size and domination properties as for the previous proof. Hence, \Cref{lem:fractional_hall} also constructs a fractional assignment for this partition and since we additionally also have the property that every voter in $N_a$ prefers $a$ to $b$ \Cref{thm:upperBound} also applies. Hence, the distortion for $a$ is at most $\max (\frac{2-k}{k}, \frac{3-k}{1-k})$. 
\end{proof}
As a corollary we get that both Bucklin's rule and the majoritarian compromise achieve a metric distortion of at most $5$ and that one can achieve a distortion of $2 + \sqrt{5}$ via the $\lceil\frac{3 - \sqrt{5}}{2}n\rceil$-approval fallback bargaining rule.
\begin{restatable}{corollary}{fallbackcor}
    Bucklin's rule, the majoritarian compromise, and thus also the COORDINATION mechanism of \citet{AFP22b} achieve a metric distortion of at most $5$. The $\lceil\frac{3 - \sqrt{5}}{2}n\rceil$-approval fallback bargaining rule achieves a metric distortion of at most $2 + \sqrt{5}$.
\end{restatable}
\begin{proof}
In the COORDINATION mechanism of \citet{AFP22b} voters (essentially) increase their approval ballots until one candidate gets approved by at least $\frac{n}{2}$ voters. Then, according to some fixed tie-breaking among the candidates, one of these candidates with at least $\frac{n}{2}$ votes is selected. It is easy to see that this indeed selects a candidate from the majoritarian compromise rule with fixed tie-breaking.

    For both Bucklin's rule and the majoritarian compromise (i.e., the strong and weak majority quota), we notice that $q-1 \le \frac{1}{2}n \le q$ holds. Thus, both have a distortion of at most $5$. Similarly for $q = \lceil\frac{3 - \sqrt{5}}{2}n\rceil$ we get $q-1 \le \frac{3 - \sqrt{5}}{2}n \le q$ and thus a distortion of at most $2 + \sqrt{5}$.
\end{proof}

\subsection{Distortion Lower Bounds}
As our next result, we show that the previously obtained bounds for $k$-plurality approval are tight. For this, we construct two families of counterexamples, one for $k \le \frac{3 - \sqrt5}{2}$ and one for $k \ge \frac{3 - \sqrt5}{2}$.
\begin{restatable}{theorem}{protolower} \label{thm:proto_lower}
For every $k \in (0,1)$ and $\varepsilon > 0$ there exists a $k$-plurality approval winner with a distortion of at least 
    $\max (\frac{2-k}{k}, \frac{3-k}{1-k}) - \varepsilon$. 
\end{restatable}
\begin{proof}
First, note that $\frac{2-k}{k} = \frac{2}{k}-1$ is a monotonically decreasing function and $\frac{3-k}{1-k} = 3 + \frac{2k}{1-k}$ is a monotonically increasing function. They intersect for $k =  \frac{3- \sqrt{5}}{2}$ so for the statement it is sufficient to find a family such that the metric distortion of a selected candidate by \Cref{protoGeneral} converges to $\frac{2-k}{k}$ for $k \leq  \frac{3- \sqrt{5}}{2}$ and one in which the metric distortion converges to $\frac{3-k}{1-k}$  for $k \geq \frac{3- \sqrt{5}}{2}$.

We start with the case that $k \geq \frac{3- \sqrt{5}}{2}$. 
Now consider the following family of instances:
Let $N = N_1 \cup N_2 \cup \{n\}$ and $C = \{c_1, \dots c_{n-2}\} \cup \{a,b\}$ with $\lvert N_1\rvert = \lceil kn \rceil -2$, and consequently $\lvert N_2\rvert =  n - \lceil kn \rceil + 1$.

We define the preferences of the voters as follows.
Each voter $i \in N_1$ has the preferences $b \succ_i a \succ_i c_1 \dots \succ_i c_{n-2}$.
Each voter $ i \in N_2$ has the preferences $c_i \succ_i b \succ_i a \succ_i c_1 \succ_i \dots \succ_i c_{i-1} \succ_i c_{i+1} \succ_i \dots \succ_i c_{n-2}$, and the final voter $n$ has the preferences $a \succ_n c_1 \succ_n \dots \succ_n c_{n-2} \succ_n b$. A visualization of the profile can be seen in  \Cref{tab:lbHalfPreliminary}.

First, we claim that $a$ is an approval winner. For this, observe that $a$ is approved by voter $n$. Furthermore, for every voter besides $n$, there are at most two candidates preferred to $a$, namely $b$ and for every voter $i \in N_2$ the candidate $c_i$.
Observe that by definition every candidate besides $b$ is the top-choice at most once and $b$ is only the top-choice for the voters in $N_1$ so precisely $\lceil kn \rceil -2$ times. Thus, it holds that no prefix without $a$ has cumulative plurality score of at least $kn$ and therefore every voter approves $a$. Therefore, $a$ is an approval winner and selected by \Cref{protoGeneral}.

In the second part we now want to argue that for $n \to \infty$ for the family of instances above there exists a corresponding family of metric spaces in which the social cost of $a$ converges against $\frac{3-k}{1-k}$ times the one of $b$.

For this given an instance with parameter $n$ consider the following metric space:

\begin{align*}
    &d(i, b)= 0, d(i,c) = 2, \text{ for all } c \in C \setminus \{b\}, i \in N_1 \\
    &d(i,b) = d(i,c_i) = 1,  d(i,c) = 3  \text{ for all } c \in C \setminus \{b,c_i\}, i \in N_2 \\
    &d(n, c) = 1 \text{ for all } c \in C. 
\end{align*}

This metric space clearly is consistent with the preferences. Furthermore, this defines a metric. Note, that this corresponds to the (0,1,2,3)-metrics defined in \citet{ChRa22a}.

Based on this we can now compute the social cost $\soc_n$ of the instance with $n$ voters. It holds that $\soc_n(b) = n -  \lceil kn \rceil + 2 \leq n -kn + 2 = (1-k)n + 2$.

On the other hand, the social cost of $\soc_n(a) = 2(\lceil kn \rceil-2) + 3 (n - \lceil kn \rceil+1) + 1  \geq 2kn -4 + 3(1-k)n + 1 \geq (3-k)n - 3$.

Thus,
$
\frac{\soc_n(a)}{\soc_n(b)} \geq  \frac{(3-k)n - 3}{(1-k)n + 2} \to \frac{3-k}{1-k}
$
for $n \to \infty$.

For the second case consider $k \leq \frac{3- \sqrt{5}}{2}$.

Let $N = N_1 \cup N_2$ and $C = \{c_1, \dots c_{n-\lceil kn\rceil -1}\} \cup \{a\}$. Further, we set $\lvert N_2\rvert = \lceil kn \rceil +1$, and consequently $\lvert N_1\rvert =  n - \lceil kn \rceil -1$.

Further, we set the preferences of the voters as follows. Each voter $i \in N_1$ has the preferences $c_{i \mod (n-\lceil kn \rceil-1)}  \succ_i \dots \succ_i c_{(n-\lceil kn \rceil-1)} \succ_i a$, and each voter $i$ in $N_2$ has the following preferences $a \succ_i c_{1} \succ_i \dots \succ_i c_{(n-\lceil kn \rceil-1)}$. A visualization of the preference profile can be seen in \Cref{tab:lbSmallK}.

First, we claim that $a$ is the unique approval winner. For this, observe that $a$ is approved by $\lceil kn\rceil + 1$ voters, namely every voter in $N_2$.  Finally, observe that each other candidate is approved at most $\lceil kn \rceil$ times.
This follows due to the following. Observe that the approval is symmetric for the voters in $N_1$ and thus each candidate gets approved by $\frac{\lceil kn \rceil}{\lvert C \setminus \{a\}\rvert} \lvert N_1 \rvert = \lceil kn \rceil$ many voters.  Further, no voter in $N_2$ approves any of these candidates since $a$ has a  plurality score of $\lceil kn \rceil + 1$.
Thus, $a$ is the unique approval winner.

For any such instance with $n$ voters consider the following metric:
\begin{align*}
&d(i, c)= 0, d(i,a) = 2, &\text{ for all } c \in C \setminus \{a\}, i \in N_1 \\
&d(i,c)= 1 &\text{ for all } c \in C, i \in N_2.
\end{align*}

Note that this metric space is consistent with the preferences. Furthermore, as seen before this defines a metric.
Next, we again want to calculate the ratio regarding the social cost of $a$ and $c_1$.
We have that
$\soc_n(a) = (\lceil kn \rceil+1) + 2(n-\lceil kn \rceil-1) \geq kn + 2n(1-k) - 2 = 2n-kn-2$ 
and
$\soc_n(c_1) =\lceil kn \rceil+1 \leq kn+2$. 

Thus,
$
\frac{\soc_n(a)}{\soc_n(c_1)} \geq  \frac{2n-kn-2}{kn+2} \to \frac{2-k}{k}
$
for $n \to \infty$. This concludes the proof.
\begin{table*}[t]
    \centering
    \begin{minipage}[t]{0.48\textwidth}
        \centering
        \begin{tabular}[t]{ccc}
            $\lceil kn \rceil-2$  & $1$ \; (each $i \in N_2$) & $1$ \\
            \hline
            $b$ & $c_i$ & $a$ \\
            $a$ & $b$ & $c_1$ \\
            $c_1$ & $a$ & $c_2$ \\
            $c_2$ & $c_1$ & $c_3$ \\
            $\vdots$  & $\vdots$ & $\vdots$  \\
            $c_i$ & $c_{i-1}$ & $c_{i+1}$ \\
            $c_{i+1}$ & $c_{i+1}$ & $c_{i+2}$ \\
            $\vdots$  & $\vdots$ & $\vdots$  \\
            $c_{n-2}$ & $c_{n-2}$ & $b$ \\
        \end{tabular}
        \caption{The instance $I_1$ establishing the bound
        $\frac{3-k}{1-k}$.
        }
        \label{tab:lbHalfPreliminary}
    \end{minipage}
    \hfill
    \begin{minipage}[t]{0.48\textwidth}
        \centering
        \begin{tabular}[t]{cc}
            $n-(\lceil kn \rceil + 1)$  & $\lceil kn \rceil + 1$ \\
            \hline
            $c_{i  \mod (n-\lceil kn \rceil-1)}$ & $a$   \\
            $\vdots$ & $c_1$  \\
            $c_{i + (n-\lceil kn \rceil-2) \mod (n-\lceil kn \rceil-1)}$ & $\vdots$   \\
            $a$  & $c_{ (n-\lceil kn \rceil-1)}$ \\
        \end{tabular}
        \caption{The instance $I_2$ establishing the bound $\frac{2-k}{k}$ for sufficiently large $n$. } 
        \label{tab:lbSmallK}
    \end{minipage}
\end{table*}

\end{proof}
We further show that one can modify these constructions to also apply to $q$-approval fallback bargaining to show the same lower bound.
\begin{restatable}{theorem}{fallbacklower} \label{thm:fallback_lower}
Let $n \in \mathbb{N}_{\ge 3}$ and $1 \le q \le n-1$. Then the distortion of $q$-approval fallback bargaining on instances with $n$ voters is at least $\max(\frac{2n - q}{q}, \frac{3n-q-1}{n-q+1})$. 
\end{restatable}
\begin{proof}
    We again provide two constructions. 

    For the first construction, we can significantly simplify the corresponding construction from \Cref{thm:proto_lower}. Our instance consists of the candidates $a, c_1, \dots, c_{n-q}$ and voter sets $N_1$ and $N_2$ with $|N_2| = q$ and $|N_1| = n-q$.
    The $q$ voters in $N_2$ all have first choice $a$ with the ranking arbitrarily completed, while for each $c_i$ (with $i \in [n-q]$) there exists a single voter in $N_1$ who ranks $c_i$ first, then all other $c$ candidates, and finally $a$. 

    It is easy to see that $a$ is indeed the $q$-approval fallback bargaining winner. Again, as in the previous proofs, we set $d(i,c_j) = 0$ and $d(i,a) = 2$ for all $j \in [n-q]$ and $i \in N_1$ and $d(i,c)= 1$ for all $c \in C$ and $i \in N_2$. 

    This shows that the distortion of $a$ is $\frac{2n - q}{q}$.

    For the other case, we mostly follow the first construction from \Cref{thm:proto_lower}. 
    We again let $N = N_1 \cup N_2 \cup \{n\}$ with $\lvert N_1\rvert = q-1$ and $\lvert N_2\rvert = n-q$. Now for each $i \in N_2$ we add \emph{two} dummy candidates $c_i^1$ and $c_i^2$, as well as our two ``main''candidates $a$ and $b$. 

    Each voter in $N_1$ ranks $b$ and then $a$ with the rest of the preference completed arbitrarily. Each voter in $N_2$ first ranks the two dummy candidates belonging to this voter, then $b$ and then $a$ and finally voter $n$ first ranks $a$, then all dummy candidates, and finally $b$. It is easy to see that $a$ is a winner at rank $2$.

    For the distances we again follow the previous proof and require that 
    \begin{align*}
    &d(i, b)= 0, d(i,c) = 2, \text{ for all } c \in C \setminus \{b\}, i \in N_1 \\
    &d(i,b) = d(i,c^1_i) =  d(i,c^2_i)= 1,  d(i,c) = 3 \\ & \quad \quad \quad \qquad \qquad\qquad \text{ for all } c \in C \setminus \{b,c_i^1,c_i^2\}, i \in N_2 \\
    &d(n, c) = 1 \text{ for all } c \in C. 
\end{align*}
again leading to a distortion of $a$ of $\frac{3n-q-1}{n-q+1}$.
\end{proof}
We note that if the ratio $\frac{q}{n}$ goes to $k$ this lower bound goes to $\max (\frac{2-k}{k}, \frac{3-k}{1-k})$.

Finally, we show that in general, no ``approval voting protocol'' using plurality polls can get a distortion of better than $2+ \sqrt{5}$. For the purpose of this result we say that an approval voting protocol is a function $f$ that for each voter $i \in N$ maps their ordinal preference and the plurality vector $\plu$ to an approval set $f(\succ_i, \plu)$. We assume that the protocol is sincere (that is for each voter it always selects a prefix of their preference) and anonymous. For us this means that for any two voters $i, j \in N$ if for every rank $r$ both voters agree on the plurality score of the candidate they put into the $r$-th place respectively (note that this does not need to be the same candidate), then $|f(\succ_i, \plu)| = |f(\succ_j, \plu)|$. That is if the preferences of the voters are indistinguishable with regard to the plurality scores, then they must approve a prefix of their preference of the same size.
\begin{restatable}{theorem}{allprotocol} 
    For every $\varepsilon > 0$ every anonymous and sincere approval voting protocol using plurality polls must have a metric distortion of at least $2 + \sqrt{5} - \varepsilon$.
\end{restatable}
\begin{proof}
    Formally, an approval voting protocol maps a vector of normalized plurality scores to a cut-off rank, with the voter then approving each candidate with at most this rank. This, particularly, means that in an election in which every candidate has the same plurality score, every voter must approve the same amount of candidates. 

    Now, consider an election with $m$ candidates and $m!$ voters such that each candidate has $(m-1)!$ first choice votes. Let $\ell$ be the number of candidates approved by each voter. Just as in \Cref{thm:proto_lower} we will construct two instances to certify our lower bound. 

    \textbf{Case 1: } Let $T_\ell$ be the set of ordered tuples of length $\ell$ over the candidates in $C$. For each $T \in T_\ell$ we add $(m-\ell)!$ voters ranking $T$ at the top of their preference list. It is easy to see that these are indeed $(m!/(m-\ell)!) \cdot (m-\ell)! = m!$ many voters and that every candidate receives the same plurality score. Further, since every voter approves $\ell$ candidates, also every candidate receives the same approval score. Let $a$ be any candidate. Now further consider the instance in which every voter who did not have $a$ in their top $\ell$ choices puts $a$ at the bottom of their preference list, while the rest of the preference lists are completed arbitrarily. 

    For this, we can consider the following line metric: every candidate except for $a$ is at $0$, $a$ is at $1$, every voter approving $a$ is at $\frac{1}{2}$ and every other voter is at $0$. As there are $\ell (m-1)!$ voters approving $a$ this leads to $a$ having a social cost of $m! - \frac{1}{2}\ell(m-1)!$ while every other candidate has a social cost of $\frac{1}{2}\ell(m-1)!$. Hence, the distortion of $a$ is 
    \[
    \frac{m! - \frac{1}{2}\ell(m-1)!}{\frac{1}{2}\ell(m-1)!} = \frac{2m - \ell}{\ell}.
    \]

    \textbf{Case 2: } This case mostly follows the first case from \Cref{thm:proto_lower}. However, now to ensure that every candidate has the same number of first choice votes we need to add clones of the $b$ voters. Thus, we choose the set of candidates to be $C = \{c_1, \dots, c_{m - \ell + 1}, a, b_1, \dots, b_{\ell-2}\}$. We denote by $\Gamma = \{c_1, \dots, c_{m - \ell + 1}\}$ the $c$ candidates and $B = \{b_1, \dots, b_{\ell-2}\}$ the $b$ candidates and add the following voters and distances:
    \begin{itemize}
        \item For each $i \in [m - \ell + 1]$ we add $(m-1)! $ voters $v_{c_i}$ with preference $c_i \succ B \succ a \succ \Gamma \setminus \{c_i\}$ with $d(v_{c_i}, c_i) = 1, d(v_{c_i}, B) = 1, d(v_{c_i}, a) = 3, d(v_{c_i}, \Gamma \setminus \{c_i\}) = 3$. 
        \item For each $i \in [\ell-2]$ we add $(m-1)!$ voters $v_{b_i}$ with preference $b_i \succ B \setminus \{b_i\} \succ a \succ \Gamma$ with $d(v_{b_i}, B) = 0, d(v_{b_i}, a) = 2, d(v_{b_i}, \Gamma) = 2$.
        \item We add $(m-1)!$ voters $v_a$ with preference $a \succ B \succ \Gamma$ with $d(v_{a}, a) = 1, d(v_{a}, B) = 1, d(v_{a}, \Gamma) = 1$. 
    \end{itemize}
    We note that every voter has $a$ among their top $\ell$ choices and is therefore a (tied) approval winner, as every voter approves it. Further, the social cost of $b_1$ is $(m-1)!(m - \ell +2)$ while the social cost of $a$ is $(m-1)! 3(m - \ell +1) + (m-1)! (\ell - 2)2 + (m-1)! = (m-1)!(3m - \ell)$ and hence the distortion of $a$ is $\frac{(3m - \ell)}{(m - \ell +2)}$. Hence, depending on the number of approvals $\ell$ the distortion of any protocol is at least $\min_{\ell \in [m]}\max(\frac{2m - \ell}{\ell},\frac{(3m - \ell)}{(m - \ell +2)})$ which goes to $2 + \sqrt5$ as $m$ goes to infinity. 
\end{proof}

\section{Robustness to Noise}
The previous section makes the assumption that we can elicit the precise plurality scores as they are in the final election. This, however, seems unrealistic in practice. In particular, polls are typically not reflective of the exact preference distribution in the population and therefore contain errors. 

Instead, we assume that polls are noisy representations of the real plurality scores. To measure the impact of such noise we study two models. First, we take a deterministic viewpoint: we assume that we are given a (not necessarily correct) \emph{prediction poll} and that the voters vote according to this poll and the $k$-plurality approval protocol. Now, we ask: what is the guaranteed worst-case metric distortion we get from the protocol as a function of the error of prediction? Formally for us, for a given set of candidates, a prediction is a probability distribution over $C$ that is a vector $q = (q(c_1), \dots, q(c_m)) \geq 0$ with $\sum_{i = 1}^m q(c_i) = 1$. Now given such a prediction $q$ and parameter $k \in (0,1)$ in the $(q,k)$-plurality approval protocol, every voter $i \in N$ approves the smallest prefix $A_i$ of their preference such that $\sum_{c \in A_i} q(c) \ge k$. 
That is, we replace the plurality scores by our prediction. Now, we evaluate the accuracy of our prediction by the distance to the original (normalized) plurality vector $p = (\plu(c_1)/n, \dots, \plu(c_m)/n)$. 
To compare the two distributions, we compare their distance on the prefixes of voters' preferences. For a given voter $i \in N$ and candidate $c \in C$ we let $U_{i,c} \coloneqq \{c' \in C\colon c' \succsim_i c\}$ denote the upper contour set of $c$ in $\succ_i$. Now the distance of $q$ to $p$ is their maximum difference among the ``probability mass'' distributed among all upper contour sets in the preference profile: $\Delta(q,p) = \max_{i \in N, c\in C} \lvert \sum_{c' \in U_{i,c}} \left(q(c') - p(c') \right) \rvert$. 

In particular, we show that if $\Delta(q,p) \le \varepsilon$ the metric distortion of the $(q,k)$-plurality approval protocol is at most  $\max\left(\frac{2-k}{k}, \frac{3 - k - \varepsilon}{1 - k - \varepsilon}\right)$. As a consequence the upper bound on the metric distortion is continuous in the error of the prediction. We visualize the resulting distortion bounds for various values of $\varepsilon$ in \Cref{fig:distbound-error}.  Our proof again follows by showing that the sets constructed from the approval relation satisfy the conditions for \Cref{lem:fractional_hall}, this time with $\delta = \varepsilon$. 
\begin{figure}[t]
    \centering
    \newcommand{\addrobustcurve}[3]{%

        \pgfmathsetmacro{\xend}{1-#1-0.001}%
        \pgfmathsetmacro{\kopt}{
            (3-#1-sqrt((#1)^2-2*#1+5))/2
        }%
        \pgfmathsetmacro{\dopt}{(2-\kopt)/\kopt}%

        \addplot[
            #2,
            very thick,
            domain=0.05:\xend
        ] {
            max(
                (2-x)/x,
                (3-x-#1)/(1-x-#1)
            )
        };
        \addlegendentry{$\varepsilon=#1$}

        \addplot[
            only marks,
            mark=*,
            mark size=1.8pt,
            #3,
            forget plot
        ] coordinates {(\kopt,\dopt)};
    }

    \begin{tikzpicture}
    \begin{axis}[
        width=\columnwidth,
        height=0.78\columnwidth,
        xlabel={$k$},
        ylabel={Distortion},
        xmin=0,
        xmax=1,
        ymin=1,
        ymax=10,
        samples=250,
        axis lines=left,
        grid=both,
        minor tick num=1,
        grid style={gray!20},
        legend style={
            at={(0.98,0.98)},
            anchor=north east,
            draw=none,
            fill=white,
            font=\small,
            row sep=-1pt,
        },
        restrict y to domain=0:10,
        unbounded coords=jump,
    ]

        \addrobustcurve{0.0}
            {blue}
            {blue}

        \addrobustcurve{0.1}
            {red!80!black, dashed}
            {red!80!black}

        \addrobustcurve{0.2}
            {teal!80!black, dashdotted}
            {teal!80!black}

        \addrobustcurve{0.4}
            {orange!90!black, densely dotted}
            {orange!90!black}

    \end{axis}
    \end{tikzpicture}

    \caption{Visualization of the distortion upper bound in \Cref{thm:error_upper} for different values of $\varepsilon$. Observe that for increasing $\varepsilon$ the optimal $k$-value decreases.}
    \label{fig:distbound-error}
\end{figure}
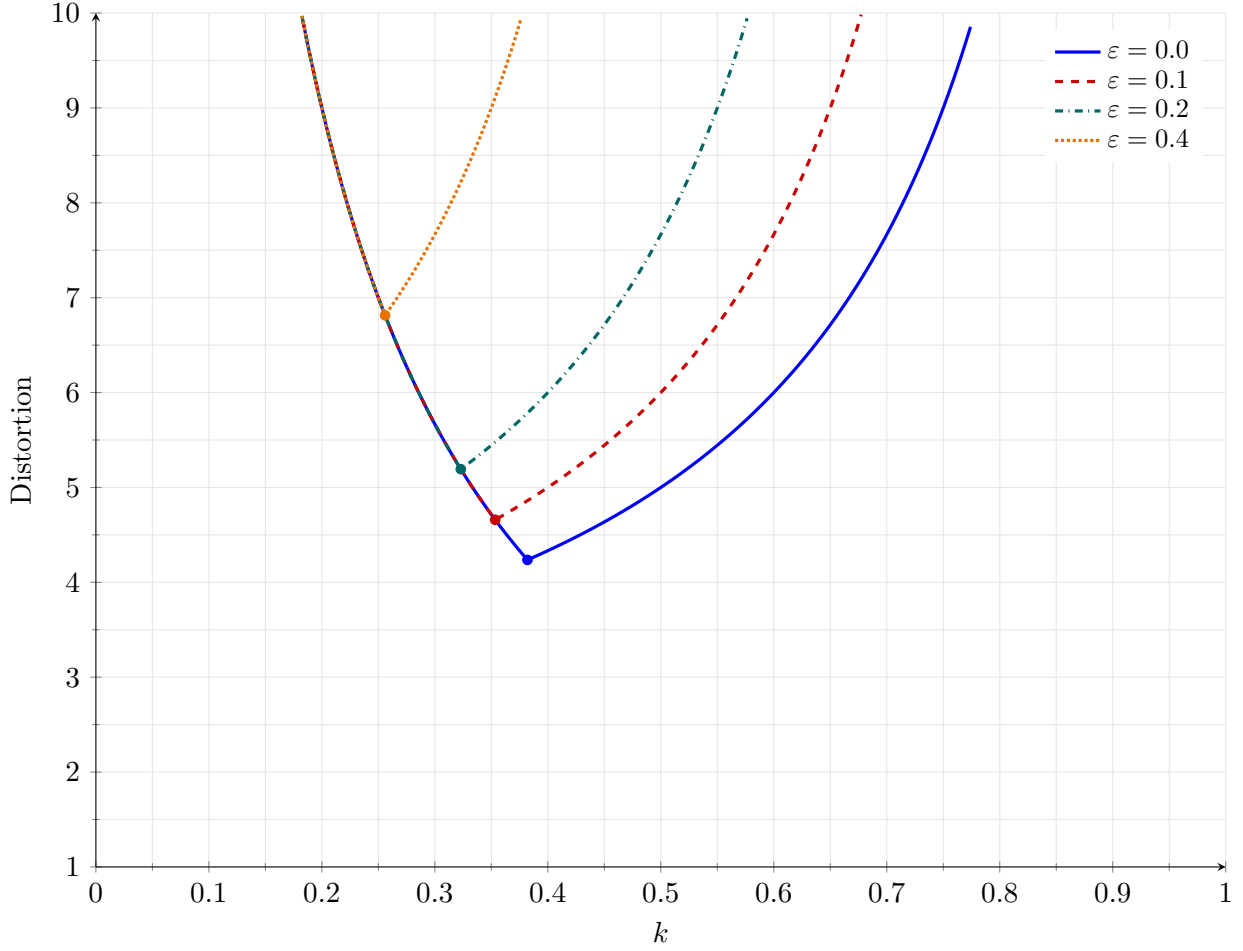

\begin{restatable}{theorem}{errorupper} \label{thm:error_upper}
    Let $q$ be a prediction poll and $\varepsilon \in [0,1)$ with $\Delta(q,p) \le \varepsilon$. Then for any $k \in (0, 1 - \varepsilon)$ the metric distortion of any $(q,k)$-plurality approval winner is at most $\max\left(\frac{2-k}{k}, \frac{3 - k - \varepsilon}{1 - k - \varepsilon}\right)$.
\end{restatable}
\begin{proof}
    We again let $N_a, N_b, N_{ab}, N_\emptyset$ as for the standard $k$-plurality approval protocol. The only property that does not immediately follow from the $k$-plurality approval is that every voter in $N_{ab}$ dominates at least $(1 - k - \varepsilon)n$ many voters. To see this, let $i \in N_{ab}$ and again $L_i = \{x \in C  \colon a \succsim_i x\}$. By definition of the protocol we know that $\sum_{c \in L_i} q(c) \ge (1-k)$. Since $\Delta(q,p) \leq \varepsilon$ this also shows that $\sum_{c \in L_i} \plu(c)/n \ge (1-k - \varepsilon)$ and thus $\sum_{c \in L_i} \plu(c) \ge (1-k - \varepsilon)n$.  There are therefore at least $(1-k - \varepsilon)n$ such voters $j$ and with $c=\top(j)$ we again get both $a\succsim_i c$ and $c\succsim_j b$. Hence, $i$ dominates at least $(1-k - \varepsilon)n$ different voters. 

    Thus, from \Cref{lem:fractional_hall} it follows that a fractional assignment exists with parameter
    \[
    t_\varepsilon  \coloneqq \max\left\{
        \frac{1-2k}{k},
        \frac{k+\varepsilon}{1-k-\varepsilon}
    \right\}.
    \]
    Thereby, we get our desired distortion upper bound.
\end{proof}

Interestingly, the error term affects only one of the two functions in our metric distortion bound. The reason is that only \Cref{lemmaOutdegree} requires that $p$ corresponds to the actual plurality scores since it uses that the plurality score corresponds to the number of voters who have a candidate $c  \in C$ as a top choice in order to lower bound the outdegree of each voter in $N_{ab}$ by $(1-k)n$. All other properties we rely on depend on the poll only through the fact that its entries sum to one.
Since the outdegree bound is what governs the second function in the maximum, only that function acquires a dependency on $\varepsilon$.

As an immediate corollary, this shows that if the total variation distance between the poll and the original normalized plurality vector is at most $\varepsilon$ the metric distortion is also upper bounded by $\max\left(\frac{2-k}{k}, \frac{3 - k - \varepsilon}{1 - k - \varepsilon}\right)$.

One canonical way to elicit such a noisy prediction $q$ is via sampling. That is, we choose a sample size $r \in [n]$, sample a set $S_r$ of $r$ voters uniformly at random and for each candidate $c \in C$ set $q(c) = |\{i \in S_r\colon \top(i) = c\}|/|S_r|$. We refer to the corresponding distribution over polls as $\mathbf{Q}(r)$. One result we immediately get from \Cref{thm:error_upper} is an upper bound on the sample size needed to achieve a distortion close to that of $k$-plurality approval via a standard application of concentration inequalities.
\begin{restatable}{theorem}{randomBounds}
    Let $\varepsilon \in  (0,1)$, $\delta \in (0,1)$ and $k \in (0,1- \varepsilon)$. Then for $r \ge \ln(2nm/\delta)/(2\varepsilon^2)$,  $r \in [n]$ with probability at least $(1 - \delta)$ the distortion of the $(q, k)$-plurality approval winner with $q$ sampled from $\mathbf{Q}(r)$ is at most $\max\left(\frac{2-k}{k}, \frac{3 - k - \varepsilon}{1 - k - \varepsilon}\right)$.
\end{restatable}
 \begin{proof}
 Given any prefix $S$ that occurs in $\succ$. We want to bound the likelihood that $\lvert q(S) -p(S) \rvert > \varepsilon$ by $\frac{\delta}{nm}$. For this, we want to use Hoeffding's inequality which also holds in the case without replacement \citep[Theorem 4]{Hoe63a}. We consider the following Bernoulli distribution. We uniformly draw $r$ voters in $N$ and have a value of $1$ if the top-choice of the selected voter lies in $S$ and $0$ otherwise. Then, we get for the empirical distribution $q$ that 
$\Pr[\lvert q(S) -p(S) \rvert \geq \varepsilon] \leq 2 \exp(-\frac{2\varepsilon^2r^2}{r}).$
 Thus, we want to show that 
$
 2 \exp(-{2\varepsilon^2r}) \leq \frac{\delta}{nm} $
 which simplifies to 
$r  \geq  \frac{1}{2\varepsilon^2}\ln(\frac{2nm}{\delta}).$

 Now using the union bound for the at most $nm$ prefixes that can occur in $\succ$ we get that with probability at most $\delta$ $\lvert q(S) -p(S) \rvert \geq \varepsilon$ for any prefix $S$. Thus, with probability at least $1- \delta$ we can apply \Cref{thm:error_upper}.
 \end{proof}
A second way to interpret the sampling process is as a randomized voting rule: that is, we sample a set $S_r$ uniformly at random, elicit the plurality distribution on the sample and then elect a winner from the respective set of approval winners. For this we fix any tie-breaking rule (the bounds we give hold for any such rule).
This gives us a randomized social choice rule that does not return a set of candidates but a probability distribution over them.

For a probability distribution $\pi$ over $C$ we define the metric distortion of $\pi$ as 
$$\dist(\pi, \succ)=\sup_{d \in \mathcal{D}(\succ)} \frac{\sum_{c\in C} \pi(c)\soc(c,d)}{\min_{e \in C}\soc(e,d)}.$$

Formally, for a given candidate $c$ let $\Pi^r_k(c)$ be the probability that $c$ is the $(q,k)$-plurality approval winner when $q$ is sampled from $\mathbf{Q}(r)$. 
In the following we provide an upper bound on the metric distortion of $\Pi^r_k$ that guarantees that for any sample size $r$ the metric distortion is at most twice our upper bound from the deterministic setting (\Cref{theoremBound}).\footnote{Note that this is different (and easier to achieve) from the expected distortion of the selected candidates.} 

\begin{restatable}{theorem}{expectedTheorem}
\label{thm:expected}
Let $\succ$ be a preference profile with $n$ voters,
  let $r \in [n]$ be the sample size, $k \in (0,1)$, and
  $t = \max\left(\frac{1-2k}{k}, \frac{k}{1-k}\right)$. Then
  $$\dist(\Pi^r_k, \succ) \le 1 + 2(t+1)\left(2 - \frac{r}{n}\right).$$
\end{restatable}
\begin{proof}
The proof is based on the following idea. If we have the full plurality poll, the proof of \Cref{lem:fractional_hall} shows how to find an allocation. We will use this matching and adapt it as follows. In order to find this matching we split the set of voters into four subsets. Then, we showed that every voter in $N_{ab} \cup N_a$ dominates every voter in $N_\emptyset$. Observe that for this statement the actual plurality poll is completely irrelevant. The same holds for the statement, that $\lvert N_b \rvert \leq \lvert N_a \rvert$. The only thing that changes is the number of voters a voter in $N_{ab} \cup N_a$ dominates. 
The proof of this statement is based on the fact that the plurality score of candidates not selected is at least $(1-k)n$. But since we only used a sample of size $r$ those values might be inflated by a factor of $\frac{n}{r}$. To compensate for this we instead try to find an allocation with the following properties. 
For this denote by $S_r$ the voters that have been polled with poll size $r$ and  by $\sigma(j) = \frac{n}{r}$ if $j \in S_r$ and $\sigma(j) = 1$ if $j \notin S_r$.
Again as before we define $t \coloneqq \max\left\{
        \frac{1-2k}{k},
        \frac{k}{1-k}
    \right\}$:

\begin{itemize}
    \item[i*.)] $\sum_{j \in N} w(i,j) = 1$ for every $i \in N_b \cup N_{ab} \cup N_\emptyset$;
    \item[ii*.)] $\sum_{i \in N_b \cup N_{ab} \cup N_\emptyset} w(i,j) \le (t + 1)\sigma(j)$ for all $j \in N \setminus N_\emptyset$;
    \item[iii*.)] $\sum_{i \in N_b \cup N_{ab} \cup N_\emptyset} w(i,j) \le (t+1)\sigma(j)-1$ for all $j \in N_\emptyset$;
    \item[iv*.)] If $w(i,j) > 0$ and $i \notin N_\emptyset$ then $i$ dominates $j$;
    \item[v*.)] If $w(i,j) > 0$ and $i \in N_\emptyset$ then $j$ dominates $i$. 
\end{itemize}

    First, for $i \in N_\emptyset$ we know that every voter in $w(i)$ dominates $i$ and hence, by \Cref{lemmaE},
    \begin{align*}
        d(i,a)
        &\le \sum_{j \in w(i)} w(i,j)
        \bigl(3d(i,b) + 2d(j,b)\bigr) \\
        &= 3d(i,b) + \sum_{j \in w(i)} w(i,j)2d(j,b).
    \end{align*}

    Second, for $i \notin N_\emptyset \cup N_a$ we know that $i$ dominates every voter in $w(i)$ and hence, by \Cref{lemmaC},
    \begin{align*}
        d(i,a)
        &\le \sum_{j \in w(i)} w(i,j)
        \bigl(d(i,b) + 2d(j,b)\bigr) \\
        &= d(i,b) + \sum_{j \in w(i)} w(i,j)2d(j,b).
    \end{align*}

        Piecing these two inequalities together we get 
    \begin{align*}
        \soc(a)
        &= \sum_{i \in N} d(i,a) \\
        &= \sum_{i \in N_a} d(i,a)
          + \sum_{i \in N_\emptyset} d(i,a) + \sum_{i \notin N_\emptyset \cup N_a} d(i,a) \\
        &\le \sum_{i \in N_a} d(i,b) \\
        &\quad + \sum_{i \in N_\emptyset}
          \left(3d(i,b) + \sum_{j \in w(i)} w(i,j)2d(j,b)\right) \\
        &\quad + \sum_{i \notin N_\emptyset \cup N_a} \left( d(i,b)
          + \sum_{j \in w(i)} w(i,j)2d(j,b) \right) \\
        &\le \sum_{i \in N_a \setminus S_r} (d(i,b) + (t+1)\cdot 2d(i,b)) \\
         &\quad + \sum_{i \in N_a \cap S_r} (d(i,b) + (t+1)\frac{n}{r}\cdot 2d(i,b)) \\
        &\quad + \sum_{i \in N_\emptyset \setminus S_r}(3d(i,b) + t \cdot 2d(i,b)) \\
        &\quad + \sum_{i \in N_\emptyset \cap S_r}(3d(i,b) +  ((t+1)\frac{n}{r}-1) \cdot 2d(i,b)) \\
        &\quad + \sum_{i \notin N_\emptyset \cup N_a \cup S_r} (d(i,b)
          + (t+1)\cdot 2d(i,b)) \\
        &\quad + \sum_{i \in (N_b \cup N_{ab}) \cap S_r} (d(i,b)
          + (t+1)\frac{n}{r}\cdot 2d(i,b)) \\
        &= \sum_{i \in  S_r} (d(i,b) + (t+1)\frac{n}{r}2d(i,b))  \\
        &\quad + \sum_{i \notin  S_r} (d(i,b)
          + (t+1)\cdot 2d(i,b)). 
    \end{align*}

Now observe that if we uniformly sample polls of size $r$ every voter gets selected with probability $\frac{r}{n}$. Therefore, for every voter $i \in N$ and a random $S_r$ it holds  $P[i \in S_r] = \frac{r}{n}$.
For this we assume that we sample every possible subset of size $r$ with equal likelihood then for each of them execute \Cref{protoGeneral} select any approval winner and denote the probability distribution over the selected candidates by $\pi$.
Next, we bound the social cost $\soc(\pi,d) =  \sum_{c \in C} \pi(c) \sum_{i \in N} d(i,c)$ for any consistent metric $d$ against any candidate $b$ as follows

\begin{align*}
    \soc(\pi, d) &= \sum_{c \in C} \pi(c) \sum_{i \in N} d(i,c) \\
            &= \sum_{\substack{S_r \in \binom{N}{r} \\ a(S_r) \text{ the selected} \\\text{winner in $S_r$}}} \frac{1}{\binom{n}{r}} \sum_{i \in N} d(i,a(S_r)) \\
            &\leq  \sum_{\substack{S_r \in \binom{N}{r}}} \frac{1}{ {\binom{n}{r}}} (\sum_{i \in S_r} (1 + 2(t+1)\frac{n}{r})d(i,b) 
            \\
            &\quad + \sum_{i \in N \setminus S_r} (1 + 2(t+1))d(i,b)) \\
            &= \sum_{i \in N} (\frac{r}{n} (1 + 2(t+1)\frac{n}{r})d(i,b) 
            \\
            &\quad + (1-\frac{r}{n})(1 + 2(t+1))d(i,b)) \\
            &=  \sum_{i \in  N}(1 + 2(2-\frac{r}{n})(t+1))d(i,b).
\end{align*}

Thus, we get a metric distortion of $1 + 2(t+1)(2-\frac{r}{n})$.

It remains to show that such an allocation always exists.
For this we use the same proof technique as above.

We claim that we can find a fractional assignment that assigns:
    \begin{itemize}
        \item voters in $N_{ab}$ to one of their at least $(1-k)r$ voters they dominate;
        \item voters in $N_b$ to voters in $N_a$;
        \item voters in $N_\emptyset$ to voters in $N_a \cup N_{ab}$.
    \end{itemize}

The proof works similarly to before. For completeness we provide the proof. 

Let $N' \subseteq N_{ab} \cup N_b \cup N_\emptyset$ be an arbitrary non-empty subset of voters. To apply Hall's Theorem it remains to show that this subset collectively has $\lvert N'\rvert$ weighted voters they could be assigned to. We will do a case distinction based on the different sets the voters could belong to. Remember that since $\sigma(j) \geq 1$, every capacity is at least as large as in the deterministic setting. We therefore only require $\sigma(j) \geq 1$ except for the case that $N' \subseteq N_b \cup N_{ab}$ with there existing at  least one voter in $N_{ab} \cap  N'$.

First, if $N' \subseteq N_b$, we know that the voters in $N'$ could be fractionally assigned to $t+1 \ge 1$ copies of $N_a$. Since $\lvert N_a\rvert \ge \lvert N_b\rvert$ the Hall condition follows.

Second, if $N' \subseteq N_b \cup N_\emptyset$ with at least one voter in $N'$ belonging to $N_\emptyset$, we know that the voters can be assigned to the $(t+1)$ copies of $N_a \cup N_{ab}$. Since $\lvert N_a \cup N_{ab}\rvert \ge kn$ we get that these are at least
    \begin{align*}
        (t+1)kn
        &= \left(\max\left(\frac{1-2k}{k},\frac{k}{1-k}\right)+1\right)kn \\
        &\ge\left(\frac{1-2k}{k}+1\right)kn \\
        &= (1-k)n \\
        &\ge \lvert N_b \cup N_\emptyset\rvert \\
        &\ge \lvert N'\rvert
    \end{align*}
    many copies.

    Next, we analyze the only case in which we have to argue differently.
    If $N' \subseteq N_b \cup N_{ab}$ with at least one voter in $N'$ belonging to $N_{ab}$, we know that the voter in $N_{ab}$ already dominates at least $(1-k)r$ many voters. Of these dominated voters, the voters in $N_\emptyset$ contribute $(t+1)\frac{n}{r}-1$ while the voters outside $N_\emptyset$ contribute $(t+1)\frac{n}{r}$. Hence, we get that these are at least \begin{align*}
        &((t+1)\frac{n}{r}-1) \lvert N_\emptyset\rvert
          + (t+1)\frac{n}{r}\bigl((1-k)r-\lvert N_\emptyset\rvert\bigr) \\
        &\quad = (t+1)(1-k)n-\lvert N_\emptyset\rvert \\
        &\quad = \left(\max\left(\frac{1-2k}{k},\frac{k}{1-k}\right)+1\right)
          (1-k)n-\lvert N_\emptyset\rvert \\
        &\quad \ge n-\lvert N_\emptyset\rvert \\
        &\quad \ge \lvert N_{ab}\rvert+\lvert N_b\rvert \\
        &\quad \ge \lvert N'\rvert .
    \end{align*}
    
   Finally, assume that $N'\subseteq N$ with there existing voters from both $N_{ab}$ and $N_\emptyset$ in $N'$. Hence, $N'$ can be assigned to all of $N_\emptyset \cup N_a\cup N_{ab}$. For this case, we need one property of $t$, namely $\max\left(\frac{1-2k}{k}, \frac{k}{1-k}\right) = t \ge 1-k$. To see this, we observe that $\frac{1-2k}{k} \ge 1-k$ if and only if $k^2 - 3k+1 \ge 0$ and $\frac{k}{1-k} \ge 1-k$ if and only if $k^2 - 3k+1 \le 0$. Thus, the inequality holds in either of the two cases.
   
   Now,  by using that $\lvert N_\emptyset\rvert \le (1-k)n$ and $\lvert N_a\cup N_{ab}\rvert \ge kn$ we get a fractional weight of at least 
   \begin{align*}
       t &\lvert N_\emptyset\rvert + (t+1)(\lvert N_a\cup N_{ab}\rvert) \ge \lvert N_a\cup N_{ab}\rvert + t\lvert N_\emptyset \cup  N_a\cup N_{ab} \rvert \\ &= \lvert N_a\cup N_{ab}\rvert + \max\left(\frac{1-2k}{k}, \frac{k}{1-k }\right)\lvert N_\emptyset \cup  N_a\cup N_{ab} \rvert \\
       & \ge \lvert N_a\cup N_{ab}\rvert + (1-k) (\lvert N_\emptyset\rvert + kn)
       \\ & \ge \lvert N_a\cup N_{ab}\rvert + (1-k)(\lvert N_\emptyset\rvert + \frac{k}{(1-k)}\lvert N_\emptyset\rvert) \\ 
       & \ge \lvert N_a\cup N_{ab}\rvert + (1-k) \left(\frac{1}{(1-k)}\lvert N_\emptyset\rvert\right) \\ 
       &\ge \lvert N_a\cup N_{ab} \cup N_{\emptyset}\rvert \ge \lvert N_b\cup N_{ab} \cup N_{\emptyset}\rvert \ge \lvert N'\rvert.
   \end{align*}
\end{proof}

We note that this bound is not tight. For instance for $r = 1$ the protocol reduces to the famous random dictatorship rule which has an expected metric distortion of less than $3$. In general, we conjecture that the distortion should be upper bounded by the deterministic bound.
\begin{conjecture}
    For any sampling size the expected metric distortion is bounded by $\max (\frac{2-k}{k}, \frac{3-k}{1-k})$.
\end{conjecture}

\section{Conclusion and Open Questions}
In this paper we introduced a simple protocol $k$-Plurality Approval for eliciting approval ballots based on plurality polls. We showed that this rule achieves a metric distortion of $2 + \sqrt{5} \simeq 4.236$ for $k = \frac{3-\sqrt{5}}{2}$ and further used our techniques to derive the metric distortion of Bucklin's rule and in general voting rules in the fallback bargaining family. We furthermore analyzed the robustness of the voting rule to errors in the poll and its distortion performance as a randomized voting rule.
As shown in \Cref{thm:query_complexity} our protocol only requires $\mathcal{O}(m)$ pairwise comparison queries per voter to find the winner of the protocol. In fact, our protocol can even be implemented as a two-round querying protocol in which, within each round, the queries posed to a voter do not depend on the queries posed to any other voter.
An interesting open question already raised by \citet{EHM24a} is whether one can achieve a constant metric distortion bound using a single round of $\mathcal{O}(m)$ pairwise comparison queries per voter, with the best currently known distortion bound in this setting being $\mathcal{O}(\log(m))$. Another interesting question raised by our work is regarding the performance of dynamics.
For instance, assume that instead of directly voting using our protocol, we elicit a second poll using it. How should voters vote given this additional information? Would perhaps a process akin to \citet{Lasl09a}'s leader rule eventually converge to alternatives with low distortion if we start off with our protocol? This also ties in with one limitation of our protocol: it is not strategyproof. In both the plurality elicitation phase and the approval phase it can be beneficial for voters to deviate from the protocol. Is it possible to quantify the impact of these deviations?

\section*{Acknowledgments}
Fabian Frank is supported by the Deutsche Forschungsgemeinschaft under grant BR 2312/14-1.

\end{document}